\documentclass[conference]{IEEEtran}
\IEEEoverridecommandlockouts

\usepackage{booktabs} 
\usepackage{framed,enumitem,amssymb}
\newlist{todolist}{itemize}{2}
\setlist[todolist]{label=$\square$}
\usepackage{pifont}

\usepackage[ruled,vlined,linesnumbered,norelsize]{algorithm2e}
\usepackage{algorithmic}
\usepackage{listings}
\usepackage{threeparttable}
\usepackage{hyperref}
\usepackage{enumitem}
\usepackage{amsmath}

\usepackage[dvipsnames]{xcolor}
\usepackage{tcolorbox}
\tcbuselibrary{skins}

\usepackage{url}
\usepackage{xspace}
\usepackage[utf8]{inputenc}
\usepackage{cuted}
\usepackage{todonotes}

\usepackage[T1]{fontenc}
\usepackage{multicol}

\usepackage[left=1in, right=1in, top=1in, bottom=1in]{geometry}

\usepackage{float}

\usepackage{amsthm}

\definecolor{codegreen}{rgb}{0,0.6,0}
\definecolor{codegray}{rgb}{0.5,0.5,0.5}
\definecolor{codepurple}{rgb}{0.58,0,0.82}
\definecolor{backcolour}{rgb}{0.95,0.95,0.99}

\usepackage[linesnumbered,ruled]{algorithm2e}
\usepackage{listings}
\usepackage{icomma}
\usepackage{graphicx}
\usepackage{pifont}

\usepackage{lipsum}

\let\OLDthebibliography\thebibliography
\renewcommand\thebibliography[1]{
  \OLDthebibliography{#1}
  \setlength{\parskip}{0pt}
  \setlength{\itemsep}{0pt plus 0.3ex}
}

\renewcommand{\textsf}{\texttt}

\newcommand{\kbcom}[1]{\textcolor{red}{[\bf KB: #1]}}

\newcommand{\remove}[1]{}

\newcommand{\algname}{\textit{MindTheStep-AsyncPSGD}}
\newcommand{\asyncsgd}{\emph{AsyncPSGD}\/}
\newcommand{\syncsgd}{\emph{SyncPSGD}\/}

\newtheorem{theorem}{Theorem}
\newtheorem{lemma}{Lemma}
\newtheorem{corollary}{Corollary}
\newtheorem{assumption}{Assumption}

\begin{document}

\title{
\algname{}: \\
\LARGE Adaptive Asynchronous Parallel Stochastic Gradient Descent
}

\author{\IEEEauthorblockN{Karl Bäckström, Marina Papatriantafilou, Philippas Tsigas}
\IEEEauthorblockA{Dept. of Computer Science and Engineering, Chalmers University of Technology, Gothenburg, Sweden\\
\{bakarl, ptrianta, tsigas\}@chalmers.se}
}

\IEEEoverridecommandlockouts
\IEEEpubid{\begin{minipage}{\columnwidth} 978-1-7281-0858-2/19/\$31.00~\copyright~2019 IEEE. Personal use of this material is permitted. Permission from IEEE must be obtained for all other uses, in any current or future media, including reprinting/republishing this material for advertising or promotional purposes, creating new collective works, for resale or redistribution to servers or lists, or reuse of any copyrighted component of this work in other works. \hfill \end{minipage} \hspace{\columnsep}\makebox[\columnwidth]{}}

\maketitle

\IEEEpubidadjcol

\begin{abstract}
Stochastic Gradient Descent (SGD) is very useful in optimization problems with high-dimensional non-convex target functions, and hence constitutes an important component of several Machine Learning and Data Analytics methods.
Recently there have been significant works on understanding the parallelism inherent to SGD, and its convergence properties. Asynchronous, parallel SGD (\asyncsgd{}) has received particular attention,
due to observed performance benefits. On the other hand, asynchrony implies inherent challenges in understanding the execution of the algorithm and its convergence, stemming from the fact that the contribution of a thread might be based on an old (stale) view of the state.
In this work we aim to deepen the understanding of \asyncsgd{}
in order to increase the statistical efficiency in the presence of stale gradients.
We propose new models for capturing the nature of the staleness distribution in a practical setting. Using the proposed models, we derive a staleness-adaptive SGD framework, \algname{}, for adapting the step size in an online-fashion, which provably reduces the negative impact of asynchrony.
Moreover, we provide general convergence time bounds for a wide class of staleness-adaptive step size strategies for convex target functions.
We also provide a detailed empirical study, showing how our approach implies faster convergence for deep learning applications.
\looseness=-1

\end{abstract}

\section{Introduction}\label{sec:introduction}

The explosion of data volumes available for Machine Learning (ML) has posed tremendous scalability challenges for machine intelligence systems. Understanding the ability to parallelise, scale and guarantee convergence of basic ML methods under different synchronization and consistency scenarios have recently attracted a significant interest in the literature. The classic 
Stochastic Gradient Descent (SGD)
algorithm is a significant target of research studying its convergence properties under parallelism. \looseness=-1 

In SGD, the goal is to minimize a function $f : \mathbb{R}^d \rightarrow \mathbb{R}$ of a $d$-dimensional vector $x$ using a first-order light-weight iterative optimization approach; i.e., given a randomly chosen starting point $x_0$, SGD repeatedly changes $x$ in the negative direction of a stochastic gradient sample, which provably is the direction in which the target function is expected to decrease the most. The step size $\alpha_t$ defines how coarse the updates are:
\begin{align}
    x_{t+1} \leftarrow x_t - \alpha_t \nabla F(x_t) \label{eq:SGD}
\end{align}
SGD is very useful in nonconvex optimization with high-dimensional target functions, and hence constitutes a major part in several ML and Data Analytics methods, such as regression, classification and clustering. In many applications, the target function is differentiable and the gradient can be efficiently computed, e.g. Artificial Neural Networks (ANNs) using Back Propagation \cite{werbos1982applications}.

To better utilize modern computing architectures, recent efforts propose 
\textit{parallel} SGD methods, complemented with different approaches for analyzing the convergence.
However, asynchrony poses challenges in understanding the algorithm due to \textit{stale} views of the state of $x$, which leads to reduced \textit{statistical efficiency} in the SGD steps, requiring a larger number of iterations for achieving similar performance. In this work, we focus on increasing the statistical efficiency of the SGD steps, and propose a staleness-adaptive framework \algname{} that adapts parameters to significantly reduce the number of SGD steps required to reach sufficient performance. Our framework is compatible with recent orthogonal works focusing on computational efficiency, such as efficient parameter server architectures \cite{ho2013more}\cite{cui2016geeps} and efficient gradient communication \cite{alistarh2017qsgd}\cite{wen2017terngrad}.

\noindent{\bf Motivation and summary of state-of-the-art}\\
Many established ML methods, such as ANN training and Regression, constitute of minimizing a function $f(x)$ that takes the form of a finite sum of error terms $L(d;x)$ parameterized by $x$, evaluated at different data points $d$ from a given set $D$ of measurements:
\begin{align}
    f_D(x) = \frac{1}{|D|}\sum_{d \in D} L(d;x) \label{eq:mini-batch_SGD}
\end{align}
where the parameter vector $x$, encodes previously gathered features from~$D$.
In this context, SGD typically selects mini-batches $B \subseteq D$ over which $f_B$ is minimized, and is known as \textit{Mini-Batch} Gradient Descent (MBGD). This type of SGD reduces the computational load in each step and hence enables processing of large datasets more efficiently. Moreover, randomly selecting mini-batches induces stochastic variation in the algorithm, which makes it effective in non-convex problems as well.

A natural approach to distribute work for objective functions of the form (\ref{eq:mini-batch_SGD}) is to utilize \textit{data parallelism}~\cite{zinkevich2010parallelized}, where different workers (threads in a multicore system or nodes in a distributed one) run SGD over different subsets of $D$. This will result in differently learned parameter vectors $x$, which are aggregated, commonly in a \emph{shared parameter server} (thread or node). The aggregation typically computes the average of the workers contributions; this approach is referred to as \textit{Synchronous} Parallel SGD (\emph{SyncPSGD}\/) due to its barrier-based nature. In its simple form, SyncPSGD has scalability issues due to the waiting time that is inherent in the aggregation when different workers compute with different speed. As more workers are introduced to the system, the waiting time will increase unbounded. Requiring only a fixed number of workers in the aggregation, known as $\lambda$-\textit{softsync}, bounds this waiting time \cite{lee2014model}.
The barrier-based nature of the synchronous approaches to parallel SGD enables a straightforward (yet expensive) linearization making the vast analysis of classical SGD applicable also to the parallel version. As a result, its convergence is well-understood also in the parallel case, which however suffers from the performance-degradation of the barrier mechanisms.
\looseness=-1

An alternative type of parallelization is \textit{Asynchronous} Parallel SGD (\asyncsgd), in which workers \textit{get} and \textit{update} the shared variable $x$ independently of each other. There are inherent benefits in performance due to that \asyncsgd{} eliminates waiting time, however the lack of coordination implies that gradients can be computed based on \textit{stale} (old) views of $x$, which are \textit{statistically inefficient}. However, gains in \textit{computational efficiency} due to parallelism and asynchrony can compensate for this, reducing the overall wall-clock computation time.

\noindent{\bf Challenges}
\asyncsgd{} shows performance benefits due to allowing workers to continue doing work independently of the progress of other workers. However, asynchrony comes with inherent challenges in understanding the execution of the algorithm and its convergence. In this work we address mainly (i) understanding the impact on the convergence and statistical efficiency of \textit{stale} gradients computed based on old views of $x$ and (ii) how to adapt the step size in SGD to accommodate for the presence of asynchrony and delays in the system.

\noindent{\bf Contributions}
With the above challenges in mind, in this work we aim to increase the understanding of \asyncsgd{} and the effect of stale gradients in order to increase the statistical efficiency of the SGD iterations. To achieve this, we find models suitable for capturing the nature of the staleness distribution in a practical setting. Under the proposed models, we derive a staleness-adaptive framework \algname{} for adapting the step size in the precense of stale gradients. We prove analytically that our framework reduces the negative impact of asynchrony. In addition, we provide an empirical study which shows that our proposed method exhibits faster convergence by reducing the number of required SGD iterations compared to \asyncsgd{} with constant step size. In some more detail:
\begin{itemize}[leftmargin=*]
    \item We prove analytically scalability limitations of the standard \syncsgd{} approach that have been observed empirically in other works.
    \item We propose a new distribution model for capturing the staleness in \asyncsgd{}, and show analytically how the optimal parameters can be chosen efficiently. We evaluate our proposed models by measuring the distance to the real staleness distribution observed empirically in a deep learning application, and compare the performance to models proposed in other works.
    \item Under the proposed distribution models, we derive efficiently computable staleness-adaptive step size functions which we show analytically can control the impact of asynchrony. We show how this enables \textit{tuning} the implicit momentum to any desired value.
    \item We provide an empirical evaluation of \algname{} using the staleness-adaptive step size function derived from our proposed model, where we observe a significant reduction in the number of SGD iterations required to reach sufficient performance.
\end{itemize}

Before the presentation of the results in Sections \ref{sec:ltd_scal}-\ref{sec:evaluation}, we outline preliminaries and background. Following the results-sections, we provide an extensive discussion on related work, conclusions and future work.\looseness=-1

\section{Preliminaries}\label{sec:preliminaries}

\subsection{Stochastic Gradient Descent}

We consider the optimization problem
\begin{equation} \label{eq:min_f_problem}
\begin{aligned}
& \underset{x}{\text{minimize}}
& & f(x)
\end{aligned}
\end{equation}
for a function $f : \mathbb{R}^d \rightarrow \mathbb{R}$.
In this context, we focus on methods to address this minimization problem (\ref{eq:min_f_problem}) using SGD, defined by (\ref{eq:SGD}) for some randomly chosen starting position $x_0$. We assume that the stochastic gradient $\nabla F$ is an unbiased estimator of $\nabla f$, i.e. $\mathbf{E}[\nabla F(x) \mid x] = \nabla f(x)$ for all $x$. This assumption holds for several relevant applications, in particular for problems of the form (\ref{eq:mini-batch_SGD}), including regression and ANN training.
We assume that the stochastic gradient samples are i.i.d, which is reasonable since the sampling occurs independently by different threads.
For the analysis in section \ref{sec:convergence} we adopt some additional standard assumptions on smoothness and convexity which we will introduce in that section. 

\subsection{System Model and Asynchronous SGD}

We consider a system with $m$ workers (that can be threads in a multicore system or nodes in a distributed one), which repeatedly compute gradient contributions based on independently drawn data mini-batches from some given data set $D$. We also consider a \emph{shared parameter server} (that can be a thread or a node respectively), which communicates with each of the workers independently, to give state information and get updates that it applies according to the algorithm it follows.

The $m$ asynchronous workers aim at performing SGD updates according to (\ref{eq:SGD}). Since each worker $W$ must get a state $x_t$ prior to computing a gradient, there can be intermediate updates from other workers before gradient from $W$ is applied. The number of such updates defines the \emph{staleness} $\tau_t$ corresponding to the gradient $\nabla F(x_t)$.

Assuming that the read and update operations can be performed atomically (see details in Section \ref{sec:method}), under the system model above, the SGD update (\ref{eq:SGD}) becomes
\begin{align}
    x_{t+1} \leftarrow x_t - \alpha_t \nabla F(v_t) \label{eq:HOGWILD}
\end{align}
where $v_t = x_{t-\tau_t}$ is the thread's \textit{view} of $x$.

We assume that the staleness values $\tau_t$ constitute a \textit{stochastic process} which is influenced by the computation speed of individual threads as well as the scheduler. Unless explicitly specified, we make no particular assumptions on the scheduler or computational speed among threads, except that all delays follow the same distribution with the same expected delay, i.e. $\mathbf{E}[\tau_t] = \bar\tau$ for all $t$. We do not require the staleness to be globally upper bounded, only that updates are eventually applied, making our system model \textit{fully asynchronous}.

While we assume above that gradient samples are pairwise independent, it is not reasonable to make the same assumption for the staleness. In fact, a staleness $\tau_t$ is by definition dependent on writing time of concurrent updates, which in turn are dependent on their respective staleness values.
For the analysis in Section \ref{sec:convergence}, we assume that \textit{stochastic gradients} and \textit{staleness} are uncorrelated, i.e. that the stochastic variation of the gradients does not influence the delays and vice versa.
This is also a realistic assumption, since delays are due to computation time and scheduling and the gradient's stochastic variation is due to random draws from a dataset.
\looseness=-1

\subsection{Momentum}
SGD is typically inefficient in \textit{narrow valleys} when the target function in some neighbourhood increases more rapidly in one direction relative to another. Such neighbourhoods are frequent in target functions that arise in ML applications due to their inherent highly irregular and non-convex nature.
Adding \textit{momentum} (\ref{eq:polyak_momentum}) to SGD has been seen to significantly improve the convergence speed for such functions. SGD with momentum, defined in (\ref{eq:polyak_momentum}), takes all previous gradient samples into account with exponentially decaying magnitude in its parameter $\mu$. As pointed out in \cite{mitliagkas2016asynchrony}, 
$\mu$ is often left out in parameter tuning, and in some instances even failed to be reported~\cite{abadi2016tensorflow}. However, the optimal value of algorithmic parameters such as $\mu$, just like $\alpha$, depends the problem, underlying hardware, as well as the choice of other parameters. Tuning $\mu$ has been shown to significantly improve performance \cite{sutskever2013importance}, especially under asynchrony~\cite{mitliagkas2016asynchrony}. \looseness=-1

For $\mu\in [0,1]$, SGD with momentum is defined by\looseness=-1
\begin{align}
    x_{t+1} \leftarrow x_t + \mu ( x_t - x_{t-1} ) - \alpha_t\nabla F(x_t) \label{eq:polyak_momentum}
\end{align}

\section{On the scalability of Sync-PSGD}
\label{sec:ltd_scal}

Optimal convergence with \syncsgd{} requires, as observed empirically in \cite{gupta2016model}, that the mini-batch size is reduced as the number of worker nodes increase. We prove analytically this empirical observation. We show that, from an optimization perspective, the effect of more workers on the convergence is equivalent to using a larger mini-batch size, which we refer to as the \textit{effective} mini-batch size. For maintaining a desired effective mini-batch size, which is the case in many applications\cite{keskar2016large}\cite{mishkin2017systematic}, workers must hence use smaller batches prior to the aggregation. Since the mini-batch size clearly is lower bounded, there is an implied strict upper bound on the number of worker nodes that can leverage the parallelization, which provides a bound on the scalability of the synchronous approach.

In mini-batch GD for target functions $f(x)$ of the form (\ref{eq:mini-batch_SGD}) the stochasticity is due to randomly drawing mini-batches $B$ of size $b$ from a dataset $D$ without replacement. For any positive mini-batch size $b$, we have that $F(x) = f_{B}(x)$ is an unbiased estimator of $f(x)$ since
\looseness=-1

\vspace{-10pt}

\begin{align*}
    &\mathbf{E}[F(x)]
    = \mathbf{E}[f_{B}(x)]
    = \frac{b}{|D|} \sum_{i} f_{B_i}(x) \\
    &= \frac{b}{|D|} \sum_{i} \frac{1}{b}\sum_{d \in B_i} L(d;x) = \frac{1}{|D|} \sum_{d \in D} L(d;x)
    = f(x)
\end{align*}

Hence, the SGD updates are in expectation representing the entire dataset $D$. Note that we assume $\bigcup B_i = D$. We have, however, that as the batch size $b$ increases, the variation of $F(x)$ diminishes. One can realize this by considering the extreme case $b=|D|$ for which the data sampling is deterministic. Hence, decreasing $b$ induces larger variance for the expectation $\mathbf{E}[F(x)] = f(x)$. This enables SGD to avoid local minima and hence be effective also in non-convex optimization problems.

The optimal value of $b$ is dependent on the problem and requires tuning. In particular, it has been seen that the convergence can suffer if $b$ is too large \cite{keskar2016large}\cite{mishkin2017systematic}.
\looseness=-1

In the following theorem we show that by increasing the number of worker nodes in SyncPSGD, from an optimization perspective, we get a behavior equivalent to a sequential execution of SGD with a larger mini-batch size, which we refer to as the \textit{effective} mini-batch size.
\looseness=-1
\begin{theorem} \label{theorem:synch_SGD_batch_problem}
    \syncsgd{} with $m$ workers, all using batch size $b$, is equivalent to a sequential execution of SGD with batch size $m \cdot b$, reffered to as effective batch size.
\end{theorem}

The proof appears 
in the appendix
due to space limitations. The main idea is to compute the average of two workers, using batch size $b$, from which it is clear that the result is equivalent to an execution of sequential SGD with batch size $2b$. The result follows inductively.\looseness=-1

Since the mini-batch size is clearly lower bounded, Theorem \ref{theorem:synch_SGD_batch_problem} implies that for a sufficiently large number of worker nodes, the effective mini-batch size scales linearly in the number of workers nodes. In order to maintain reasonable mini-batch size with sufficient variation in the updates, this implies a strict upper bound on the number of workers nodes. Moreover, under the assumption that there is an optimal mini-batch size $b^*$ for a given problem, which has been seen to be a common assumption, we have that the maximum number of workers possible in order to achieve optimal convergence is exactly $m=b^*$, each using mini-batch size $b=1$.

\section{The proposed framework} \label{sec:method}

We outline \algname{} for staleness-adaptive steps and analyze how to choose a suitable adaptive step size function under different staleness models.
Due to space limitations, proofs appear in
the appendix, 
while brief arguments are presented here instead. \looseness=-1

\subsection{The \algname{} Framework}
 We consider a standard parameter-server type of algorithm \cite{ho2013more}\cite{li2014scaling}, with atomic read and write operations, ensuring that workers acquire consistent views of the state $x$. In a distributed system, the consistency can be realized through the communication protocol. In a multi-core system, where worker nodes are threads and $x$ can be stored on shared memory, consistency can be realized with appropriate synchronization and producer-consumer data structures, with the extra benefit that they can pass pointers to the data (parameter arrays) instead of moving it.
In Algorithm \ref{algorithm:AdAsyncSGD} we show the pseudocode for \algname{}, describing how standard \asyncsgd{} using a parameter server (thread or node) is extended with a staleness-adaptive step.

\SetKwBlock{Repeat}{repeat}{}

\begin{algorithm} \label{algorithm:AdAsyncSGD}
\begin{small}
    GLOBAL start point $x_0$, functions $F(x)$ and $\alpha(\tau$) \\
    
    \vspace{5pt}
    
    \begin{multicols}{2}
    
    \vspace{3pt}
    
    \underline{Worker $W$}\;
    $(t, x) \leftarrow (0, x_0)$ \\
    \Repeat{
        compute $g \leftarrow \nabla F(x)$ \\
        \textit{send} $(t, g)$ to $S$ \\
        \textit{receive} $(t, x)$ from $S$
    }
    
    \vspace{15pt}
    
    \underline{Parameter server $S$}\;
    $(t', x) \leftarrow (0, x_0)$ \\
    \Repeat{
        \textit{receive} $(t, g)$ from a ready worker $W$ \\
        $\tau \leftarrow t' - t$ \\
        $x \leftarrow x - \alpha(\tau) g$ \\
        $t' \leftarrow t' + 1$ \\
        \textit{send} $(t', x)$ to $W$
    }
    \end{multicols}
\end{small}
\vspace{8pt}
\caption{\algname{}}
\end{algorithm}

Note that \algname{} as a framework essentially  ``modularizes" the role of $\alpha$ as a parameter that can configure and tune performance, with criteria and benefits that are analysed in the next subsection.

\subsection{Tuning the impact of asynchrony}
As pointed out in \cite{mitliagkas2016asynchrony}, asynchrony and delays introduce \textit{memory} in the behaviour of the algorithms. In particular, in~\cite[Theorem~2]{mitliagkas2016asynchrony}, they quantify this and show its resemblance to momentum, however for a constant step size. The corresponding result for a stochastic staleness-adaptive step size is formulated here:
\looseness=-1

\begin{lemma} \label{lemma:async_update_sum}
    Let $\tau$ be distributed according to some PDF $p$ such that $P[\tau=i]=p(i)$. Then, for an adaptive step size function $\alpha(\tau)$, we have
    \begin{align}
    \begin{split}
        \mathbf{E}[x_{t+1} - x_t] = \mathbf{E}[x_t - x_{t-1}] + \sum_{i=0}^{\infty} \big( p(i) \alpha (i) - \\
        p(i+1) \alpha (i+1) \big) \nabla f(x_{t-i-1}) - p(0) \alpha(0)\nabla f(x_t)
    \end{split}
    \end{align}
\end{lemma}
The proof of Lemma \ref{lemma:async_update_sum} follows the structure the one in \cite{mitliagkas2016asynchrony}, now taking into account the adaptive step size.
The main {takeaways} from Lemma \ref{lemma:async_update_sum} are that, under asynchrony, (i)~the gradient contribution diminishes as the number of workers increases\footnote{Here it is assumed that $p(0)$ tends to zero as the number of workers increases. This is easily realized for our proposed $CMP$ $\tau$ model (\ref{eq:cmp_distribution}). For the geometric staleness model we confirm empirically in section \ref{sec:evaluation} that this assumption holds in practice, recall that $p(0) = p$.\label{fn:decaying_p0}}; (ii)~there is a momentum-like term introduced with parameter $\mu = 1$ and (iii)~the update depends on the series term:
\looseness=-1
\begin{align}
    \Sigma_{p,\alpha}^\nabla = \sum_{i=0}^{\infty} \big( p(i) \alpha (i) - p(i+1) \alpha (i+1) \big) \nabla f(x_{t-i-1}) \label{eq:sigma_term}
\end{align}
which quantifies the potential impact of stale gradients depending on the distribution of $\tau$.

The issue of diminishing gradient contributions as the number of workers increase can in theory be resolved by choosing a larger $\alpha$. However, this would require step sizes proportional to $p(0)^{-1}$, which rapidly grows out of bounds as the number of workers increase. Since large $\alpha$ can significantly impact the statistical efficiency of the SGD steps in practice and in fact needs to be carefully tuned, this poses a scalability limitation.

This is where \algname{} can help tune the impact of asynchrony, as we show in the following.\looseness=-1

\textbf{Momentum from geometric $\tau$.} Assuming a geometrically distributed $\tau$, the series $\Sigma_{p,\alpha}^\nabla$ is manifested in the convergence behaviour in the form of asynchrony-induced memory with a \textit{momentum} effect; see Theorem 3  of \cite{mitliagkas2016asynchrony}, repeated here for self-containment:

\begin{theorem}[\cite{mitliagkas2016asynchrony}] \label{theorem:implicit_momentum}
     Let all $\tau_t$ be geometrically distributed with parameter $p$, i.e. $\mathbf{P}[\tau=k] = p(1-p)^k$. Then, for a constant $\alpha$, the expected update (\ref{eq:HOGWILD}) becomes
    \begin{align}
        \mathbf{E}[x_{t+1} - x_t] = (1-p) \mathbf{E}[x_t - x_{t-1}] - p\alpha\nabla f(x_t) \label{eq:impllicit_momentum}
    \end{align}
\end{theorem}

The statement of Theorem \ref{theorem:implicit_momentum} is easily confirmed by substituting $p(i)$ in (\ref{eq:sigma_term}) with constant $\alpha$ with the geometric PDF, which yields $\Sigma_{p,\alpha}^\nabla = - p \mathbf{E}[x_t - x_{t-1}]$.

Eq. (\ref{eq:impllicit_momentum}) resembles the definition of momentum, with expected implicit asynchrony-induced momentum of magnitude $\mu = 1-p$. As the number of workers grow and $p$ tends to $0$, Theorem \ref{theorem:implicit_momentum} suggests an implicit momentum that approaches $1$. This would imply a scalability limitation since the parameter $\mu$ requires careful tuning.

Assuming a geometric staleness model, we show in the following theorem how \algname{} with a particular step size function resolves this issue.

\begin{theorem} \label{theorem:implicit_momentum_tuning}
    Let  staleness $\tau \in \text{Geom}(p)$ and
    \begin{align}
        \alpha_t = C^{-\tau_t}p^{-1} \alpha \label{equation:momentum_alpha}
    \end{align}
    where $\alpha$ is a parameter to be chosen suitably. Then
    \begin{align*}
        \mathbf{E}[x_{t+1} - x_t] = \mu_{C,p} \mathbf{E}[x_t - x_{t-1}] - \alpha\nabla f(x_t)
    \end{align*}
    and the implicit asynchrony-induced momentum is
    \begin{align}
        \mu_{C,p} = 2 - (1-p)/C
    \end{align}
\end{theorem}

This is confirmed by substituting $\alpha(\tau)$ in (\ref{eq:sigma_term}) with the adaptive step(\ref{equation:momentum_alpha}). Note that the expected implicit momentum vanishes for $C=(1-p)/2$. More generally: \looseness=-1

\begin{corollary}
    Any desired momentum $\mu^*$ is, in expectation, implicitly induced by asynchrony by using the staleness-adaptive step size in (\ref{equation:momentum_alpha}) with
    \begin{align}
        C = (1-p)/(2-\mu^*)
    \end{align}
\end{corollary}

\textbf{Applicability of geometric $\boldsymbol\tau$.} Each gradient staleness is comprised by two parts, one of which is the staleness $\tau_C$ which counts the number of gradients applied from other workers concurrent with the gradient computation. The second part of the staleness, which we denote $\tau_S$, counts, after the gradient computation of a worker finishes, the number of gradients from other workers which are applied first, which is decided by the order with which the workers are scheduled to apply their updates. The complete staleness of a gradient is $\tau = \tau_C + \tau_S$. Note that, if we assume a uniform fair stochastic scheduler, then $\tau_S$ is decided exactly by the number of Bernoulli trials until a specific gradient is chosen, hence $\tau_S \in \text{Geom}(\cdot)$. The geometric $\tau$ model is  therefore applicable for problems where $\tau_C << \tau_S$, i.e. when the gradient computation time typically is smaller than the time it takes to apply a computed gradient (eq.~\ref{eq:HOGWILD}).

Now consider also relevant applications of SGD where  the gradient computation time $\tau_C$ is far from negligible, e.g the increasingly popular Deep Learning, which typically includes ANN training with BackProp~\cite{werbos1982applications} for gradient computation. The BackProp algorithm requires in the best case multiple multiplications of matrices of dimension $d$, which by far dominates the SGD update step (\ref{eq:HOGWILD}) which consists of exactly $d$ floating point multiplications and additions. For such applications the geometric $\tau$ model is hence not sufficient; we confirm this empirically in Section \ref{sec:evaluation}. In the following, we propose a class of $\tau$ distributions which is more suitable.

\textbf{Conway-Maxwell-Poisson (CMP) $\boldsymbol\tau$.} Considering applications with time-consuming gradient computation such as ANN training, we aim to find a suitable staleness model. Since now we consider (i)~that $\tau_C >> \tau_S$ and (ii)~that applying a computed gradient is relatively fast,
we can consider the completion of gradient computations as rare arrival events. This opts for a variant of the Poisson distribution, such as the CMP distribution which in addition to Poisson has a parameter $\nu$ which controls the rate of decay. We have that $\tau \in \text{CMP}(\lambda, \nu)$ if
\begin{align}
    P[\tau = i] = \frac{1}{Z(\lambda,\nu)} \frac{\lambda^i}{ (i!)^\nu } \ , \ 
    Z(\lambda,\nu) = \sum_{j=0}^\infty \frac{\lambda^i}{(j!)^\nu} \label{eq:cmp_distribution}
\end{align}
which reduces to the Poisson distribution in the special case $\nu = 1$, i.e if $\tau \in \text{CMP}(\lambda, 1)$ then $\tau \in \text{Poi}(\lambda)$.
For the remainder of this section we aim to further investigate the behaviour of parallelism in SGD under the CMP and Poisson models, and propose an adaptive step size strategy to reduce the negative impact and improve the statistical efficiency under asynchrony.

In a homogeneous system with $m$ equally powerful worker nodes/threads, we expect that the most frequent staleness observation (the distribution mode) should relate to the number of workers. More precisely, since a sequential execution would always have $\tau=0$, an appropriate choice of $\tau$ distribution should have the mode $m-1$. For the CMP distribution, we have that if $\tau \in \text{CMP}(\lambda, \nu)$ then the mode of $\tau$ is $\lfloor \lambda^{1/\nu} \rfloor$, and we therefore hypothesize the following relation:
\begin{align}
    \lambda^{1/\nu} = m \label{eq:lambda_nu_m}
\end{align}

For the special case $\nu=1$, i.e. a Poisson $\tau$ model, (\ref{eq:lambda_nu_m}) enables us to immediately choose an appropriate value for $\lambda$ given the number of workers $m$. In general, (\ref{eq:lambda_nu_m}) simplifies the parameter search when fitting a CMP distribution model to a one-dimensional line search, which is in practice a significant complexity reduction.

\textbf{$\boldsymbol \tau$-adaptive $\boldsymbol \alpha$.} In the following, we argue analytically about how to choose an adaptive step size function $\alpha(\tau)$ for reducing the negative impact of stale gradients. We will see how a certain $\tau$-adaptive step size can bound the magnitude of $\Sigma_{p,\alpha}^\nabla$ (\ref{eq:sigma_term}), and even tune the implicit asynchrony-induced momentum to any desired value.
\begin{theorem} \label{theorem:cmp_sigma_zero}
    Assume $\tau \in CMP(\lambda,\nu)$, and let the adaptive step size function be defined as follows:
    \begin{align}
        \alpha(\tau) = C \lambda^{-\tau}(\tau!)^\nu \alpha \label{eq:adap_alpha_cmp_sigma_zero}
    \end{align}
    for any constant $C$. Then we have $\Sigma_{p,\alpha}^\nabla = 0$.
\end{theorem}
Theorem \ref{theorem:cmp_sigma_zero} shows how a simple and tunable $\tau$-adaptive step size mitigates the $\Sigma_{p,\alpha}^\nabla$ quantity.
The proof consists of confirming that each contribution of the sum $\Sigma_{p,\alpha}^\nabla$ (\ref{eq:sigma_term}) vanishes when applying the definition of the CMP distribution (\ref{eq:cmp_distribution}) and the adaptive step size (\ref{eq:adap_alpha_cmp_sigma_zero}).

However, from Lemma \ref{lemma:async_update_sum}, we see that even though $\Sigma_{p,\alpha}^\nabla$ is mitigated by the adaptive step size (\ref{eq:adap_alpha_cmp_sigma_zero}), the SGD steps still have a fixed implicit momentum term of magnitude $\mu=1$. We show in Theorem \ref{theorem:cmp_sigma_tune} how the implicit momentum can be tuned to any desired value through a particular choice of $\alpha(\tau)$.
\begin{theorem} \label{theorem:cmp_sigma_tune}
    Assume $\tau \in CMP(\lambda,\nu)$. Then, we have that $\Sigma_{p,\alpha}^\nabla$ in expectation takes the form of asynchrony-induced momentum of magnitude exactly $K$, i.e.
    \begin{align*}
        \Sigma_{p,\alpha}^\nabla = K \mathbf{E}[x_t - x_{t-1}]
    \end{align*}
    when using the adaptive step size function:
    \begin{align}
        \alpha(\tau) = c(\tau) \lambda^{-\tau}(\tau!)^\nu \alpha \label{eq:adap_alpha_cmp_sigma_tune}
    \end{align}
    where
    \begin{align}
        c(\tau) = 1 - \frac{K}{\alpha e^\lambda} \sum_{j=0}^{\tau-1} \frac{\lambda^j}{(j!)^\nu} \label{eq:cmp_adap_factor}
    \end{align}
\end{theorem}

Theorem \ref{theorem:cmp_sigma_tune} shows how the series term $\Sigma_{p,\alpha}^\nabla$ can take the form of momentum of desired magnitude by using a particular $\tau$-adaptive step size. The main idea of the proof is the observation that the contributions of $\Sigma_{p,\alpha}^\nabla$ are simplified by the particular choice of the adaptive factor $c(\tau)$ (\ref{eq:adap_alpha_cmp_sigma_tune}), and the result follows from the definition of expectation. The $c(\tau)$ contains a sum that is $\mathcal{O}(\tau)$ in computation time. This indicates that such an adaptive step size function might not scale well, since $\tau$ is expected to be in the magnitude of $m$. In the following Corollary we show how this is resolved by the corresponding $\alpha(\tau)$ under the Poisson $\tau$-model.

\begin{corollary} \label{cor:Pois_sigma_tune}
    Assuming $\tau \in Pois(\lambda)$, the series term $\Sigma_{p,\alpha}^\nabla$ takes the form of implicit momentum of magnitude $K$ when using the adaptive step size function:
    \begin{align}
        \alpha(\tau) = \left( 1 - \frac{K}{\alpha} \frac{\Gamma(\tau,\lambda)}{\Gamma(\tau)} \right) \lambda^{-\tau} \tau! \alpha \label{eq:poisson_alpha}
    \end{align}
    where $\Gamma(\cdot)$ and $\Gamma(\cdot,\cdot)$ are the Gamma and Upper Incomplete Gamma function, respectively.
\end{corollary}

Corollary \ref{cor:Pois_sigma_tune} shows how the series $\Sigma_{p,\alpha}^\nabla$ is in expectation replaced by momentum of any desired magnitude. Assuming Poisson $\tau$, the $\mathcal{O}(\tau)$ sum in (\ref{eq:cmp_adap_factor}) is replaced in (\ref{eq:poisson_alpha}) by the Gamma and Upper Incomplete Gamma function, for which there exist efficient ($\mathcal{O}(1)$) and accurate numerical approximation methods \cite{greengard2019algorithm}.

\section{Convex convergence analysis}\label{sec:convergence}

In this section we analyze the convergence time of \algname{}-type algorithms for convex and smooth optimization problems.
Here, too, due to space limitations, the proofs appear
the appendix, 
while brief intuitive arguments are presented instead.

Consider the optimization problem (\ref{eq:min_f_problem}) where an acceptable solution $x$ satisfies $\epsilon$-convergence, defined as
\begin{align}
    \|x-x^*\|^2 \leq \epsilon
\end{align}

We assume that the problem is addressed using \algname{} under the system model described in Section \ref{sec:preliminaries}. Note that we consider a staleness-adaptive step size, hence $\alpha_t=\alpha(\tau_t)$ is stochastic.

For the analysis in this section, we consider strong convexity and smoothness, specified in \textit{Assumption \ref{assumption:function_smoothness}}. These analytical requirements are common in convergence analysis for convex problems \cite{de2015taming}\cite{sra2015adadelay}\cite{alistarh2018podc}\cite{alistarh2018nips}.

\begin{assumption} \label{assumption:function_smoothness}
    We assume that the objective function $f$ is, in expectation with respect to the stochastic gradients, strongly convex with parameter $c$ with L-Lipschitz continuous gradients and that the second momentum of the stochastic gradient is upper bounded.
        \begin{gather}
            \mathbf{E} \left[ (x-y)^T\big(\nabla f(x) - \nabla f(y)\big) \mid x,y \right] \geq c \| x - y \|^2 \label{condition:convexity} \\
            \mathbf{E} \left[ \|\nabla F(x) - \nabla F(y)\|] \mid x,y \right] \leq L\| x - y \| \label{condition:lipschitz} \\
            \mathbf{E} \left[ \|\nabla F(x)\|^2 \mid x \right] \leq M^2 \label{condition:moment}
        \end{gather}
\end{assumption}

The assumption (\ref{condition:convexity}) is standard in convex optimization and ensures that gradient-based methods will converge to a global optimum. Lipschitz continuity (\ref{condition:lipschitz}) is a type of strong continuity which bounds the rate with which the gradients can vary. Due to that $\mathbf{E}[\nabla F(x^*)]=0$, (\ref{condition:moment}) can be interpreted as bounding the variance of the gradient norm around the optimum $x^*$.

In addition to our system model in Section \ref{sec:preliminaries}, we make the following assumption on the staleness process:

\begin{assumption} \label{assumption:tau_mean_independence}
    The staleness process $(\tau_i)$ is non-anticipative, i.e. mean-independent of the outcome of future states of the algorithm (e.g. future delays and gradients). In particular, we have
    \begin{align*}
        \mathbf{E} [\tau_i \mid \tau_t] = \mathbf{E} [\tau_i] \text{ for all } i < t
    \end{align*}
\end{assumption}

Assumption \ref{assumption:tau_mean_independence} is reasonable considering that the staleness (i.e. scheduler's decisions) at time $i$ should not be considered to be influenced by staleness values $\tau_t$ of gradients yet to be computed.

Under Assumptions \ref{assumption:function_smoothness} and \ref{assumption:tau_mean_independence} above, we give a general bound on the number of iterations sufficient for expected $\epsilon$-convergence in the following theorem:

\begin{theorem} \label{theorem:convex_convergence}
    Consider the unconstrained convex optimization problem of (\ref{eq:min_f_problem}). Under Assumptions \ref{assumption:function_smoothness} and \ref{assumption:tau_mean_independence}, for any $\epsilon > 0$, there is a sufficiently large number $T$ of asynchronous SGD updates of the form (\ref{eq:HOGWILD}) such that
    \begin{align}
    \begin{split}
        T \leq \bigg( &2 \big( c - L M \epsilon^{-1/2} \mathbf{E}\left[ \tau \alpha \right] \big) \mathbf{E}\left[\alpha\right] - \\
        &\epsilon^{-1}M^2 \mathbf{E}\left[\alpha^2\right] \bigg)^{-1} \ln{(\|x_0-x^*\|^2\epsilon^{-1}) } \label{eq:convex_convergence}
    \end{split}{}
    \end{align}
    for which we have $\mathbf{E} [\|x_T - x^*\|^2] < \epsilon$
\end{theorem}

The main idea in the proof of Theorem \ref{theorem:convex_convergence} is to bound $|| x_{t+1} - x^* || / || x_{t} - x^* ||$, which quantifies the improvement of each SGD step. The statement then follows from a recursive argument.

\begin{corollary} \label{corollary:constant_alpha_convergence}
    Consider the optimization problem and the conditions as in Theorem~\ref{theorem:convex_convergence}. There exists a choice of a step size $\alpha$ such that the convergence time $T$ is in the magnitude of $\mathcal{O}\left(\mathbf{E}\left[\tau \right]\right)$ (remember $\mathbf{E}\left[\tau \right]$ is denoted by $\bar\tau$). In particular, letting $\alpha$ be
    \begin{align}
        \alpha = \theta \frac{c \epsilon M^{-1}}{M + 2 L \sqrt{\epsilon} \bar\tau} \label{eq:alpha_choice}
    \end{align}
    for a tunable factor $\theta \in (0,2)$, there exists a $T$ such that
    \begin{align}
        T \leq \frac{M + 2 L \sqrt{\epsilon} \bar\tau}{\theta(2-\theta) c^2 M^{-1} \epsilon} \ln(\epsilon^{-1} \|x_0-x^*\|^2) \label{eq:theta_convergence}
    \end{align}
\end{corollary}

The results in Theorem \ref{theorem:convex_convergence} and Corollary \ref{corollary:constant_alpha_convergence} are related to the results presented in \cite{de2015taming} and \cite{alistarh2018podc}. The main differences are that in our analysis we tighten the bound with a factor $(2-\theta)^{-1}$, expand the allowed step size interval, as well as relax the \textit{maximum staleness} assumption and reduce the magnitude of the bound from linear in the \textit{maximum} staleness $\mathcal{O}(\hat\tau)$ to the \emph{expected} $\mathcal{O}(\bar\tau)$.

In the following corollary, we give a general bound assuming \textit{any} non-increasing step size function $\alpha(\tau)$.

\begin{corollary} \label{corollary:decaying_alpha_convergence}
    Under the same conditions as Theorem \ref{theorem:convex_convergence}, let $\alpha_t = \alpha(\tau_t)$ be a non-increasing function of $\tau_t$. Then we have the following bound on the expected number of iterations until convergence:
    \begin{align}
    \begin{split}
        T \leq &\big( 2 c \mathbf{E}\left[\alpha \right] -\epsilon^{-1}M \left( M + 2L\sqrt{\epsilon} \bar\tau \right) \mathbf{E}\left[\alpha^2\right] \big)^{-1} \\
        &\cdot \ln(\epsilon^{-1} \|x_0-x^*\|^2) \label{eq:decaying_alpha_convergence}
    \end{split}
    \end{align}
\end{corollary}

Corollary \ref{corollary:decaying_alpha_convergence} describes a general convergence bound for any step size function $\alpha(\tau)$ which decays in $\tau$. We see that such step size functions also achieve the asymptotic $\mathcal{O}(\bar\tau^{-1})$ bound, similar to the one for a constant $\alpha$ (\ref{eq:theta_convergence}).

\remove{

In the following corollary, we see how a simple $\times \tau^{-1}$ adaptive step size can recover the convergence (\ref{eq:theta_convergence}), however this requires scaling up $\alpha$ with $\bar\tau$. \kbcom{Maybe even skip this cor? One of the related works suggests the 1/tau adaptive alpha for the softsync SGD protocol. One good point with this cor is that we see that to recover the same convergence bound, we need to scale up all the alphas with $\bar\tau$, giving rise to convergence problems for small tau. But since the adaptive alpha we propose is different from this, it could be confusing.}

\begin{corollary} \label{corollary:tau_adapt_concergence}
    Under the same conditions as Theorem \ref{theorem:convex_convergence}, using the staleness-adaptive step size
    \begin{align}
        \alpha = \frac{1}{\tau_t} \theta \frac{c \epsilon M^{-1}}{M + 2 L \sqrt{\epsilon} \bar\tau} \bar\tau =: \frac{1}{\tau_t} \theta \nu \bar\tau \label{eq:alpha_choice_linear_adaptive}
    \end{align}
    where
    \begin{align}
        0 < \theta < 2 \cdot \min \bigg( 1, \frac{M/\bar\tau + 2L\sqrt{\epsilon}}{3M + L\sqrt{\epsilon}} \bigg) \label{eq:theta_restriction}
    \end{align}
    
    Then, convergence is expected within
    \begin{align}
        T \leq \frac{M + 2 L \sqrt{\epsilon} \bar\tau}{\theta(2-\theta) c^2 M^{-1} \epsilon} \ln(\epsilon^{-1} \|x_0-x^*\|^2) \label{eq:theta_convergence_linear_adaptive}
    \end{align}
\end{corollary}

}

\section{Experimental study}\label{sec:evaluation}

In this section we aim to evaluate the results derived in section \ref{sec:method} in a practical setting. This is achieved by (i)~measuring the accuracy and scalability of the proposed $\tau$-models (ii)~evaluating the convergence properties of \algname{} with an adaptive step size function derived under the CMP/Poisson $\tau$ models.

\begin{figure}
  \centering
  \includegraphics[width=0.9\linewidth]{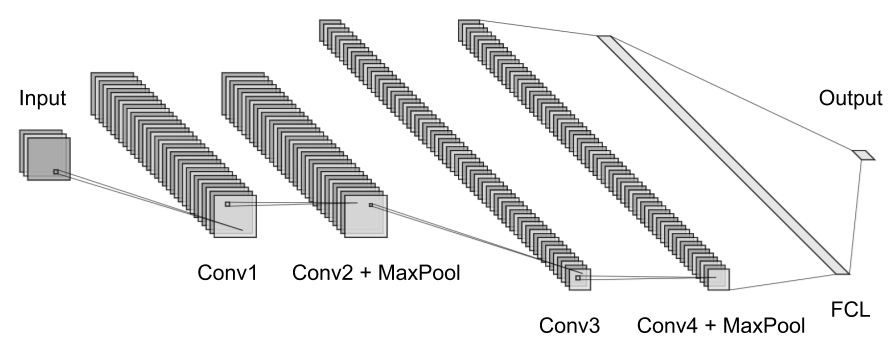}
  \caption{CNN architecture; Four convolutional layers with $3 \times 3$ kernels, with intermediate MaxPool layers. First two convolutions have 32 filters, the last two 64. The architecture has two fully connected layers, one with $256$ neurons, and the output layer with $10$ neurons.}
  \label{fig:CNNarch}
\end{figure}

\textbf{Setup.} We apply \algname{} for training a 4-layer Convolutional Neural Network (CNN) architecture (see Fig. \ref{fig:CNNarch}) on the common image classification benchmark dataset CIFAR10~\cite{krizhevsky2009learning}. The performance of the CNN is measured as the \textit{cross entropy} between the true and the predicted class distribution. The algorithm is evaluated on a setup with a 36-thread Intel Xeon CPU and 64GB memory. The implementation is in Python 2.7 and uses the standard Python multiprocessor library as well as TensorFlow \cite{abadi2016tensorflow} for gradient computation.

\textbf{CMP/Poisson $\boldsymbol\tau$.} We evaluate the $\tau$ models (Poisson, CMP) proposed in section \ref{sec:method} by comparing with the $\tau$ distribution observed in practice
for different number of workers. We compare our proposed $\tau$ models with distributions proposed in other works, namely the geometric $\tau$ model \cite{mitliagkas2016asynchrony} and the bounded uniform $\tau$ model~\cite{sra2015adadelay}.

The distribution parameters in Table~\ref{tab:dist_parameter_search} are found through an exhaustive search where we aim to minimize the Bhattacharyya distance to the $\tau$ distribution observed in practice.
Note that: 
(i) For the Poisson $\tau$ model, as hypothesized in Section \ref{sec:method}, the distribution parameter $\lambda$ indeed corresponds well to the number of worker nodes. From Fig.~\ref{fig:model_dist_fit} 
we see that the proposed CMP and Poisson $\tau$ models by far outperforms the geometrical and uniform $\tau$ models, in particular for larger number of workers. 
(ii)~As mentioned in Footnote~\ref{fn:decaying_p0},
we confirm in Table~\ref{tab:dist_parameter_search} that $\textbf{P}[\tau=0]$, i.e. $p$, decays as $m$ increases. 
(iii)~We see in Fig.~\ref{fig:model_dist_fit} that the CMP $\tau$ model outperforms the others in terms of accuracy and scalability. The CMP distribution parameter $\nu$ is found through a 1-d search, and using the assumption (\ref{eq:lambda_nu_m}) the other parameter $\lambda$ is calculated. The result in Fig. \ref{fig:model_dist_fit} therefore validates the assumption (\ref{eq:lambda_nu_m}).

\begin{figure}
  \includegraphics[width=0.9\linewidth]{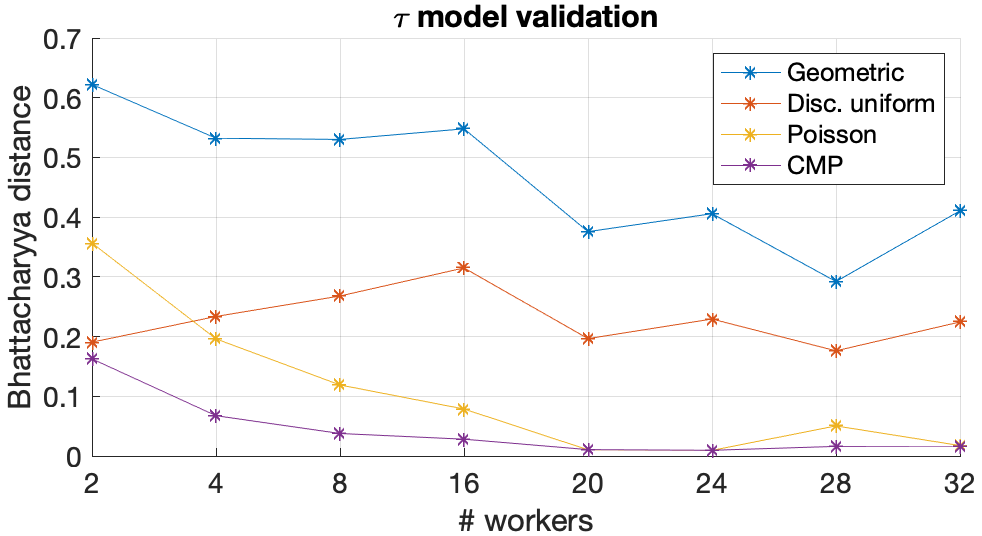}
  \caption{Bhattacharyya distance of different $\tau$ models compared to the observed distribution. The graph shows that the CMP $\tau$ model is the most accurate in all tests, with the Poisson $\tau$ model as a close second. The uniform and geometric $\tau$ models are persistently less accurate, and show poor scalability in comparison.}
  \label{fig:model_dist_fit}
\end{figure}

\begin{table}
    \centering
    \resizebox{\columnwidth}{!}{
    \begin{tabular}{ |c||c|c|c|c|c|c|c|c| }
        \hline
        $\tau$ model & 2 & 4 & 8 & 16 & 20 & 24 & 28 & 32 \\
        \hline
        $p$ (Geom) & 0.34 & 0.21 & 0.12 & 0.06 & 0.05 & 0.04 & 0.04 & 0.03 \\
        $\hat\tau$ (Unif) & 2 & 5 & 11 & 22 & 31 & 37 & 48 & 48 \\
        $\lambda$ (Pois) & 2.0 & 4.0 & 8.0 & 16.0 & 19.7 & 23.8 & 26.5 & 32 \\
        $\nu$ (CMP) & 6.28 & 5.26 & 4.18 & 3.48 & 0.93 & 0.95 & 0.39 & 0.87 \\
        \hline
    \end{tabular}
    }
    \caption{Optimal distribution parameters for different number of workers}
    \label{tab:dist_parameter_search}
\end{table}

\textbf{Convergence with $\boldsymbol\tau$-adaptive $\boldsymbol\alpha$.} We evaluate \algname{} compared with standard \asyncsgd{} by measuring the number of \emph{epochs} required until a certain error threshold is reached, epochs being the number of passes through the dataset. The number of SGD iterations in one epoch is $\lceil |D|/b \rceil$ where $|D|$ is the size of the dataset and $b$ the batch size. In our experiments we have $\lceil |D|/b \rceil = 469$.
We consider performance in terms of \textit{statistical efficiency}, i.e. the statistical benefit of each SGD step. In practice, the approach can be applied to any orthogonal work focusing on computational efficiency, such as efficient parameter server architectures \cite{ho2013more}\cite{cui2016geeps} and efficient gradient communication and quantization \cite{alistarh2017qsgd}\cite{wen2017terngrad}.

We compare standard \asyncsgd{} with constant step size $\alpha_c = 0.01$ to \algname{} with an adaptive step size function according to (\ref{eq:poisson_alpha}) with $\alpha=\alpha_c$, $K=1$, and $\lambda = m$. In addition, we bound the step size $\alpha(\tau) \leq 5 \cdot \alpha_c$ to mitigate issues with numerical instability in the SGD algorithms, and (very infrequent) gradients with $\tau>150$ are not applied.

In principle, given a sufficiently small $\alpha_c$, speedup can always be achieved by using an adaptive step size strategy $\alpha(\tau)$ which overall increases the average step size. To ensure a fair comparison, the adaptive step size function $\alpha(\tau)$ is normalized so that
\begin{align}
    \textbf{E}_\tau[\alpha(\tau)] = \alpha_c \label{eq:scale_exp_alpha}
\end{align}
where the expectation is taken over the real $\tau$ distribution observed in the system. Enforcing (\ref{eq:scale_exp_alpha}) ensures that any potential speedup is achieved due to how the step size function $\alpha(\tau)$ adaptively changes the impact of gradients depending on their staleness, and not because of the overall magnitude of the step size.

Fig. \ref{fig:convergence_comparison_split} shows how \algname{} exhibits persistently faster convergence for different number of workers. For many workers ($m=28,32$) \algname{} requires significantly fewer epochs compared to standard \asyncsgd{} to achieve sufficient performance. Observe that for $m=32$ the average speedup is $\times 1.5$ while the worst-case is $\times 1.7$.

\begin{figure}
  \centering
  \includegraphics[width=0.9\linewidth]{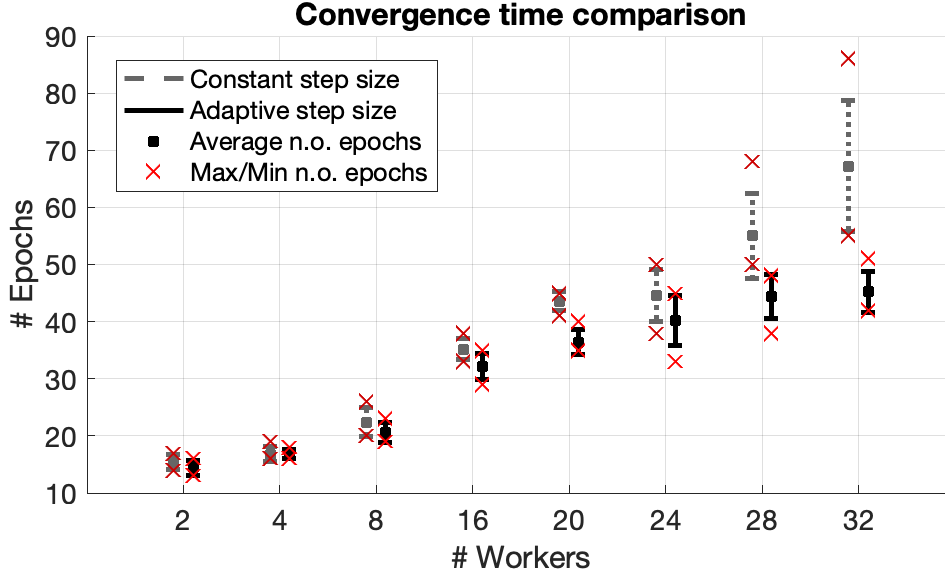}
  \caption{\asyncsgd{} vs. \algname{} comparison. The plot shows the n.o. epochs required until sufficient performance (cross-entropy loss $\leq 0.05$). The statistics are computed based on 5 runs, and the bar height corresponds to the standard deviation.}
  \label{fig:convergence_comparison_split}
\end{figure}
\section{Related work}\label{sec:relatedwork}

Orthogonal to this work, there are numerous works dedicated to optimizing the effectiveness of SGD by utilizing data sparsity, topology of the search space, and other properties of the problems. One example is introducing momentum to the updates, originally proposed in \cite{polyak1964some}, however not in the context of SGD. Apart from this, there are several variation of SGD in the sequential case introducing adaptivness to aspects of the problem topology, such as Adagrad, Adadelta, RMSprop, Adam, AdaMax, and Nadam (cf.~\cite{ruder2016overview} and references therein).

In \cite{mitliagkas2016asynchrony} Mitliagkas et al. show that under certain stochastic delay models, asynchrony has an effect on convergence similar to momentum, referred to as asynchrony-induced or implicit momentum, where more workers imply a larger magnitude of the effect. In \cite{liu2018towards} these similarities are investigated further, and it is shown that \asyncsgd{} and momentum shows different convergence rates in general and that \asyncsgd{} is in fact faster in expectation. Since it has been seen \cite{sutskever2013importance} that the magnitude of momentum can have significant impact on convergence, the result by Mitliagkas et al. would imply a harsh scalability limitation of AsyncPSGD. In this paper, we show that under the same $\tau$ model as in \cite{mitliagkas2016asynchrony}, \algname{} can in theory mitigate this issue, and even allow the expected asynchrony-induced momentum to be tuned implicitly by the rate of adaptation. In addition, in this work we propose a different class of $\tau$ distribution models, and show how they better capture the real $\tau$ values observed in a deep learning application. From our proposed models we derive an adaptive step size function $\alpha(\tau)$ which we show significantly reduces the number of SGD steps required for convergence. \looseness=-1

Below we give a brief overview of works on synchronous distributed SGD. Under smoothness and convexity assumptions, in \cite{zinkevich2010parallelized} and \cite{agarwal2010distributed}, synchronous distributed SGD with data-parallelism was observed and proven to accelerate convergence. This was implemented on a larger scale by Dekel et al. \cite{dekel2011optimal} where the convergence rates were improved under stronger analytical assumptions.
In \cite{ho2013more} and \cite{lee2014model} the synchronization is relaxed using a Stale Synchronous Parameter Server with a tunable staleness threshold in order to reduce the waiting-time, which is shown to outperform synchronous SGD. In \cite{gupta2016model} Gupta et al. give a rigorous empirical investigation of practical trade-offs the number of workers, mini-batch size and staleness; the results provide useful insights in scalability limitations in synchronous methods with averaging. We address this issue in this paper from a theoretical standpoint and explain the results observed in practice. This is discussed in detail in Section \ref{sec:ltd_scal}.

The study of numerical methods under parallelism is not new, and sparked due to the works by Bertsekas and Tsitsiklis \cite{bertsekas1989parallel} in 1989. Recent works \cite{chaturapruek2015asynchronous}\cite{lian2015asynchronous} show under various analytical assumptions that the convergence of Async-PSGD is not significantly affected by asynchrony and that the noise introduced by delays is asymptotically negligible compared to the noise from the stochastic gradients.
This is confirmed in \cite{chaturapruek2015asynchronous} for convex problems (linear and logistic regression) for a small number of cores.
In \cite{lian2015asynchronous} Lian et al. relax the theoretical assumptions and establish convergence rates for non-convex minimization problems, assuming bounded gradient delays and number of workers.
Lock-free Async-PSGD in shared-memory, i.e. \textsc{Hogwild}!, was proposed by Niu et al. \cite{recht2011hogwild} and was shown to achieve near-optimal convergence rates assuming sparse gradients. Properties of Async-PSGD with sparse or component-wise updates have since been rigourously studied in recent literature due to the performance benefits of lock-freedom \cite{sallinen2016high}\cite{nguyen2018sgd}\cite{de2015taming}. The gradient sparsity assumption was relaxed in the recent work \cite{alistarh2018podc} which magnified the convergence time bound in the order of magnitude $\sim\sqrt{d}$, $d$ being the problem dimensionality.

Delayed optimization in completely asynchronous first-order optimization algorithms was analyzed initially in~\cite{agarwal2011distributed}, where Agarwal et al. introduce step sizes which diminish over the progression of SGD, depending on the maximum staleness allowed in the system, but not adaptive to the actual delays observed.
In comparison, in this work we relax the maximum staleness restriction and derive a strategy for adapting the step size depending on the actual staleness values observed in the system in an online fashion. 
Adaptiveness to delayed updates during execution was proposed and analyzed in \cite{mcmahan2014delay} under assumptions of gradient sparsity and \textit{read} and \textit{write} operations having the same relative ordering. A similar approach was used in \cite{Zhang:2016:SAD:3060832.3060950}, however for synchronous SGD with the \textit{softsync} protocol. In \cite{Zhang:2016:SAD:3060832.3060950} statistical speedup is observed in some cases for a limited number of worker nodes, however by using \textit{momentum SGD}, which is not the case in their theoretical analysis, and step size decaying schedules on top of the staleness-adaptive step size. \looseness=-1

The work closest to ours is AdaDelay \cite{sra2015adadelay} which addresses a particular constrained convex optimization problem, namely training a logistic classifier with projected gradient descent. It utilizes a network of worker nodes computing gradients in parallel which are aggregated at a central parameter server with a step size that is scaled proportionally to $\tau^{-1}$. The staleness model in \cite{sra2015adadelay} is a uniform stochastic distribution, which implies a strict upper bound on the delays, making the system partially asynchronous. In comparison, in this work we  analyze the convergence of \algname{} for non-convex optimization, relax the bounded gradient staleness assumption, as well as evaluate more  delay models both theoretically and empirically. Moreover, we validate our findings experimentally by training a Deep Neural Network (DNN) classifier using real-world dataset, which constitutes a highly non-convex and high-dimensional optimization problem. In addition, we provide convergence analysis in the convex case for \algname{}, where we show explicitly a probabilistic time bound for $\epsilon$-convergence, for any step size function decaying in the staleness $\tau$.

\section{Conclusions and Future Work}\label{sec:conclusions}

In this paper, we first analytically confirm scalability limitations of the standard \syncsgd{}, which were observed empirically in other works; we thus motivate the need to further investigate asynchronous approaches. We propose a new class of $\tau$-distribution models, show analytically how the parameters can be efficiently chosen in a practical setting, and validate the models empirically, as well as compare to models proposed in other works.

We derive and analyse adaptive step size strategies which reduce the impact of asynchrony and stale gradients, using our framework \algname{}. We show that the proposed strategies enable turning asynchrony into implicit asynchrony-induced momentum of desired magnitude. We provide convergence bounds for a wide class of $\tau$-adaptive step size strategies for convex target functions. We validate our findings empirically for a deep learning application and show that \algname{} with our proposed step size strategy converges significantly faster compared to standard \asyncsgd{}. \looseness=-1

The concept of staleness-adaptive \asyncsgd{} has been under-explored, despite that, as shown here, it significantly improves scalability and helps maintain statistical efficiency. Continuing to investigate 
asynchrony-aware SGD, is therefore of interest. Future research directions also include further studying the nature of the staleness, i.e. effect of schedulers and synchronization methods, for understanding the impact of asynchrony and for choosing appropriate adaptive strategies.

\small
\section{Acknowledgements}
This work was partially supported by the Wallenberg AI, Autonomous Systems and Software Program (WASP), Knut and Alice Wallenberg Foundation,  
the SSF  proj. ``FiC'' nr. {GMT14-0032} and the Chalmers Energy AoA framework proj. INDEED. \looseness=-1

\begin{small}
\bibliographystyle{abbrv}
\bibliography{ref}  

\begin{thebibliography}{10}

\bibitem{abadi2016tensorflow}
M.~Abadi, P.~Barham, J.~Chen, Z.~Chen, A.~Davis, J.~Dean, M.~Devin,
  S.~Ghemawat, G.~Irving, M.~Isard, et~al.
\newblock Tensorflow: a system for large-scale machine learning.
\newblock In {\em OSDI}, volume~16, pages 265--283, 2016.

\bibitem{agarwal2011distributed}
A.~Agarwal and J.~C. Duchi.
\newblock Distributed delayed stochastic optimization.
\newblock In {\em Advances in Neural Information Processing Systems}, pages
  873--881, 2011.

\bibitem{agarwal2010distributed}
A.~Agarwal, M.~J. Wainwright, and J.~C. Duchi.
\newblock Distributed dual averaging in networks.
\newblock In {\em Advances in Neural Information Processing Systems}, pages
  550--558, 2010.

\bibitem{alistarh2018podc}
D.~Alistarh, C.~De~Sa, and N.~Konstantinov.
\newblock The convergence of stochastic gradient descent in asynchronous shared
  memory.
\newblock In {\em ACM Symp. on Principles of Distributed Computing}, PODC '18,
  pages 169--178. ACM, 2018.

\bibitem{alistarh2017qsgd}
D.~Alistarh, D.~Grubic, J.~Li, R.~Tomioka, and M.~Vojnovic.
\newblock Qsgd: Communication-efficient sgd via gradient quantization and
  encoding.
\newblock In {\em Advances in Neural Information Processing Systems}, pages
  1709--1720, 2017.

\bibitem{alistarh2018nips}
D.~Alistarh, T.~Hoefler, M.~Johansson, N.~Konstantinov, S.~Khirirat, and
  C.~Renggli.
\newblock The convergence of sparsified gradient methods.
\newblock In {\em Advances in Neural Information Processing Systems}, pages
  5977--5987, 2018.

\bibitem{bertsekas1989parallel}
D.~P. Bertsekas and J.~N. Tsitsiklis.
\newblock {\em Parallel and distributed computation: numerical methods},
  volume~23.
\newblock Prentice hall Englewood Cliffs, NJ, 1989.

\bibitem{chaturapruek2015asynchronous}
S.~Chaturapruek, J.~C. Duchi, and C.~R{\'e}.
\newblock Asynchronous stochastic convex optimization: the noise is in the
  noise and sgd don't care.
\newblock In {\em Advances in Neural Inf. Processing Systems}, 2015.

\bibitem{cui2016geeps}
H.~Cui, H.~Zhang, G.~R. Ganger, P.~B. Gibbons, and E.~P. Xing.
\newblock Geeps: Scalable deep learning on distributed gpus with a
  gpu-specialized parameter server.
\newblock In {\em Proc. of the 11th European Conf. on Computer Systems},
  page~4. ACM, 2016.

\bibitem{de2015taming}
C.~M. De~Sa, C.~Zhang, K.~Olukotun, C.~R\'{e}, and C.~R\'{e}.
\newblock Taming the wild: A unified analysis of hogwild-style algorithms.
\newblock In {\em Advances in Neural Information Processing Systems 28}, pages
  2674--2682. Curran Associates, Inc., 2015.

\bibitem{dekel2011optimal}
O.~Dekel, R.~Gilad-Bachrach, O.~Shamir, and L.~Xiao.
\newblock Optimal distributed online prediction.
\newblock In {\em Proc. of the 28th Int'l Conf. on Machine Learning (ICML-11)},
  pages 713--720, 2011.

\bibitem{greengard2019algorithm}
P.~Greengard and V.~Rokhlin.
\newblock An algorithm for the evaluation of the incomplete gamma function.
\newblock {\em Advances in Computational Mathematics}, 45(1):23--49, 2019.

\bibitem{gupta2016model}
S.~Gupta, W.~Zhang, and F.~Wang.
\newblock Model accuracy and runtime tradeoff in distributed deep learning: A
  systematic study.
\newblock In {\em 16th Int'l Conf. on Data Mining (ICDM),}, pages 171--180.
  IEEE, 2016.

\bibitem{ho2013more}
Q.~Ho, J.~Cipar, H.~Cui, S.~Lee, J.~K. Kim, P.~B. Gibbons, G.~A. Gibson,
  G.~Ganger, and E.~P. Xing.
\newblock More effective distributed ml via a stale synchronous parallel
  parameter server.
\newblock In {\em Advances in neural information processing systems}, pages
  1223--1231, 2013.

\bibitem{keskar2016large}
N.~S. Keskar, D.~Mudigere, J.~Nocedal, M.~Smelyanskiy, and P.~T.~P. Tang.
\newblock On large-batch training for deep learning: Generalization gap and
  sharp minima.
\newblock {\em arXiv:1609.04836}, 2016.

\bibitem{krizhevsky2009learning}
A.~Krizhevsky, G.~Hinton, et~al.
\newblock Learning multiple layers of features from tiny images.
\newblock Technical report, Citeseer, 2009.

\bibitem{lee2014model}
S.~Lee, J.~K. Kim, X.~Zheng, Q.~Ho, G.~A. Gibson, and E.~P. Xing.
\newblock On model parallelization and scheduling strategies for distributed
  machine learning.
\newblock In {\em Advances in neural information processing systems}, pages
  2834--2842, 2014.

\bibitem{li2014scaling}
M.~Li, D.~G. Andersen, J.~W. Park, A.~J. Smola, A.~Ahmed, V.~Josifovski,
  J.~Long, E.~J. Shekita, and B.-Y. Su.
\newblock Scaling distributed machine learning with the parameter server.
\newblock In {\em 11th Symp. on Operating Systems Design and Implementation},
  pages 583--598, 2014.

\bibitem{lian2015asynchronous}
X.~Lian, Y.~Huang, Y.~Li, and J.~Liu.
\newblock Asynchronous parallel stochastic gradient for nonconvex optimization.
\newblock In {\em Advances in Neural Information Processing Systems}, pages
  2737--2745, 2015.

\bibitem{liu2018towards}
T.~Liu, S.~Li, J.~Shi, E.~Zhou, and T.~Zhao.
\newblock Towards understanding acceleration tradeoff between momentum and
  asynchrony in nonconvex stochastic optimization.
\newblock In {\em Advances in Neural Information Processing Systems 31}, pages
  3686--3696. Curran Associates, Inc., 2018.

\bibitem{mcmahan2014delay}
B.~McMahan and M.~Streeter.
\newblock Delay-tolerant algorithms for asynchronous distributed online
  learning.
\newblock In {\em Advances in Neural Information Processing Systems 27}, pages
  2915--2923. Curran Associates, Inc., 2014.

\bibitem{mishkin2017systematic}
D.~Mishkin, N.~Sergievskiy, and J.~Matas.
\newblock Systematic evaluation of convolution neural network advances on the
  imagenet.
\newblock {\em Computer Vision and Image Understanding}, 161:11--19, 2017.

\bibitem{mitliagkas2016asynchrony}
I.~Mitliagkas, C.~Zhang, S.~Hadjis, and C.~R{\'e}.
\newblock Asynchrony begets momentum, with an application to deep learning.
\newblock In {\em Communication, Control, and Computing (Allerton), 2016 54th
  Annual Allerton Conf. on}, pages 997--1004. IEEE, 2016.

\bibitem{nguyen2018sgd}
L.~M. Nguyen, P.~H. Nguyen, M.~van Dijk, P.~Richt{\'a}rik, K.~Scheinberg, and
  M.~Tak{\'a}{\v{c}}.
\newblock Sgd and hogwild! convergence without the bounded gradients
  assumption.
\newblock {\em arXiv:1802.03801}.

\bibitem{polyak1964some}
B.~T. Polyak.
\newblock Some methods of speeding up the convergence of iteration methods.
\newblock {\em USSR Computational Mathematics and Mathematical Physics},
  4(5):1--17, 1964.

\bibitem{recht2011hogwild}
B.~Recht, C.~Re, S.~Wright, and F.~Niu.
\newblock Hogwild: A lock-free approach to parallelizing stochastic gradient
  descent.
\newblock In {\em Advances in Neural Information Processing Systems (NIPS) 24},
  pages 693--701. Curran Associates, Inc., 2011.

\bibitem{ruder2016overview}
S.~Ruder.
\newblock An overview of gradient descent optimization algorithms.
\newblock {\em arXiv:1609.04747}, 2016.

\bibitem{sallinen2016high}
S.~Sallinen, N.~Satish, M.~Smelyanskiy, S.~S. Sury, and C.~R{\'e}.
\newblock High performance parallel stochastic gradient descent in shared
  memory.
\newblock In {\em Parallel and Distributed Processing Symp., 2016 IEEE Int'l},
  pages 873--882. IEEE, 2016.

\bibitem{sra2015adadelay}
S.~Sra, A.~W. Yu, M.~Li, and A.~Smola.
\newblock Adadelay: Delay adaptive distributed stochastic optimization.
\newblock In {\em Proc. of the 19th Int'l Conf. on Artificial Intelligence and
  Statistics}, pages 957--965, 2016.

\bibitem{sutskever2013importance}
I.~Sutskever, J.~Martens, G.~Dahl, and G.~Hinton.
\newblock On the importance of initialization and momentum in deep learning.
\newblock In {\em Int'l Conf. on machine learning}, pages 1139--1147, 2013.

\bibitem{wen2017terngrad}
W.~Wen, C.~Xu, F.~Yan, C.~Wu, Y.~Wang, Y.~Chen, and H.~Li.
\newblock Terngrad: Ternary gradients to reduce communication in distributed
  deep learning.
\newblock In {\em Advances in neural information processing systems}, pages
  1509--1519, 2017.

\bibitem{werbos1982applications}
P.~J. Werbos.
\newblock Applications of advances in nonlinear sensitivity analysis.
\newblock In {\em System modeling and optimization}, pages 762--770. Springer,
  1982.

\bibitem{Zhang:2016:SAD:3060832.3060950}
W.~Zhang, S.~Gupta, X.~Lian, and J.~Liu.
\newblock Staleness-aware async-sgd for distributed deep learning.
\newblock In {\em Proceedings of the Twenty-Fifth Int'l Joint Conf. on
  Artificial Intelligence}, IJCAI'16, pages 2350--2356. AAAI Press, 2016.

\bibitem{zinkevich2010parallelized}
M.~Zinkevich, M.~Weimer, L.~Li, and A.~J. Smola.
\newblock Parallelized stochastic gradient descent.
\newblock In {\em Advances in neural information processing systems}, pages
  2595--2603, 2010.

\end{thebibliography}
\end{small}

\newpage
\appendix

\begin{proof}[Proof of Theorem \ref{theorem:synch_SGD_batch_problem}]
    Consider the case with two worker nodes. Assuming that the batches are disjoint, which is likely for large datasets, each SGD step is of the form
    \begin{align*}
        x_{t+1} &= \frac{ \big(x_t - \alpha \nabla f_{B_1}(x_t)\big) + \big(x_t - \alpha \nabla f_{B_2}(x_t)\big) }{2} \\
        &= x_t - \frac{\alpha}{2} \left( \nabla f_{B_1}(x_t) + \nabla f_{B_2}(x_t)\right)
    \end{align*}
    For mini-batch GD, i.e. a target function of the form (\ref{eq:mini-batch_SGD}), and with mini-batch size $b$, the above formula  becomes: 
    \begin{align*}
        x_{t+1} &= w - \frac{\alpha}{2} \bigg( \nabla \frac{1}{b} \sum_{d \in B_1} L(d, x_t) + \nabla \frac{1}{b} \sum_{d \in B_2} L(d, x_t) \bigg)
    \end{align*}
    From linearity of the gradient, we have
    \begin{align*}
        x_{t+1} &= w - \alpha \nabla \frac{1}{2b} \sum_{d \in B_1 \cup B_2} L(d, x_t) \\
        &= w - \alpha \nabla f_{B_1 \cup B_2}(x_t)
    \end{align*}
    that corresponds to the SGD step with batch size $2b$. This inductively implies the theorem statement.
\end{proof}

\medskip

\begin{proof}[Proof of Theorem \ref{theorem:implicit_momentum_tuning}]
    We have from (\ref{eq:HOGWILD})
    \begin{align*}
        x_{t+1} - x_t
        &= - \alpha_t\nabla F(v_t) \\
        &= x_t - x_{t-1} - (x_t - x_{t-1}) - \alpha_t\nabla F(v_t) \\
        &= x_t - x_{t-1} + \alpha_t\nabla F(v_{t-1}) - \alpha_t\nabla F(v_t)
    \end{align*}
    Since the gradient and staleness processes are independent, we take first expectation conditioned on the staleness
    \begin{align*}
        \mathbf{E}[x_{t+1} - x_t \mid \tau_t, \tau_{t-1}]
        &= \mathbf{E}[x_t - x_{t-1} \mid \tau_t, \tau_{t-1}] \\
        &+ \alpha_t\nabla f(v_{t-1}) - \alpha_t\nabla f(v_t)
    \end{align*}
    Now, take expectation w.r.t. the stochastic staleness $\tau_t, \tau_{t-1}$
    \begin{align*}
        \mathbf{E}[&x_{t+1} - x_t]
        = \mathbf{E}[x_t - x_{t-1}] \\
        & \ \ + \mathbf{E}[\alpha_t\nabla f(v_{t-1})] - \mathbf{E}[\alpha_t\nabla f(v_t)] \\
        &= \mathbf{E}[x_t - x_{t-1}] + \sum_{i=0}^\infty P[\tau=i] \frac{\alpha\nabla f(x_{t-i-1})}{C^{i}p} \\
        & \ \ - \sum_{i=0}^\infty P[\tau=i] \frac{\alpha\nabla f(x_{t-i})}{C^{i}p} \\
        &= \mathbf{E}[x_t - x_{t-1}] + p \sum_{i=0}^\infty (1-p)^i \frac{\alpha\nabla f(x_{t-i-1})}{C^{i}p} \\
        & \ \ - p \sum_{i=0}^\infty (1-p)^i \frac{\alpha\nabla f(x_{t-i})}{C^{i}p} \\
        &= \mathbf{E}[x_t - x_{t-1}] - \alpha\nabla f(x_{t}) + \sum_{i=0}^\infty (1-p)^i \frac{\alpha\nabla f(x_{t-i-1})}{C^{i}} \\
        & \ \ - \sum_{i=1}^\infty (1-p)^i \frac{\alpha\nabla f(x_{t-i})}{C^{i}}
    \end{align*}
    \begin{align*}
        &= \mathbf{E}[x_t - x_{t-1}] - \alpha\nabla f(x_{t}) \\
        & \ \ + \sum_{i=0}^\infty \bigg( \frac{(1-p)^i}{C^{i}} - \frac{(1-p)^{i+1}}{C^{i+1}} \bigg) \alpha\nabla f(x_{t-i-1}) \\
        &= \mathbf{E}[x_t - x_{t-1}] - \alpha\nabla f(x_{t}) \\
        & \ \ + \sum_{i=0}^\infty \frac{(1-p)^i}{C^{i}} \bigg( 1 - \frac{1-p}{C} \bigg) \alpha\nabla f(x_{t-i-1}) \\
        &= \mathbf{E}[x_t - x_{t-1}] - \alpha\nabla f(x_{t}) \\
        & \ \ + \bigg( 1 - \frac{1-p}{C} \bigg) \sum_{i=0}^\infty \frac{p(1-p)^i}{C^{i}p^{i+1}} \alpha\nabla f(x_{t-i-1}) \\
        &= \mathbf{E}[x_t - x_{t-1}] - \alpha\nabla f(x_{t}) \\
        & \ \ + \bigg( 1 - \frac{1-p}{C} \bigg) \mathbf{E}[\alpha_t\nabla f(v_{t-1})] \\
        &= \mathbf{E}[x_t - x_{t-1}] - \alpha\nabla f(x_{t}) + \bigg( 1 - \frac{1-p}{C} \bigg) \mathbf{E}[x_t - x_{t-1}] \\
        &= \bigg( 2 - \frac{1-p}{C} \bigg) \mathbf{E}[x_t - x_{t-1}] - \alpha\nabla f(x_{t})
    \end{align*}
\end{proof}

\begin{proof}[Proof of Theorem \ref{theorem:cmp_sigma_zero}]
    We have
    \begin{align*}
        \Sigma_{p,\alpha}^\nabla
        &= \sum_{i=0}^{\infty} \big( p(i) \alpha (i) - p(i+1) \alpha (i+1) \big) \nabla f(x_{t-i-1})
    \end{align*}
    Substituting $p(i)$ for the $CMP$ PDF (\ref{eq:cmp_distribution}) gives
    \begin{align}
        \Sigma_{p,\alpha}^\nabla
        &= \frac{1}{Z(\lambda,\nu)} \sum_{i=0}^{\infty} \frac{\lambda^i}{(i!)^\nu} \left( \alpha (i) - \lambda \frac{\alpha (i+1)}{(i+1)^\nu} \right) \nabla f(x_{t-i-1}) \label{eq:sigma_CMP}
    \end{align}
    Now, applying the adaptive step size (\ref{eq:adap_alpha_cmp_sigma_zero}) gives
    \begin{align*}
        \Sigma_{p,\alpha}^\nabla
        = &\frac{C}{Z(\lambda,\nu)} \sum_{i=0}^{\infty} \frac{\lambda^i}{(i!)^\nu} \alpha \bigg( \lambda^{-i}(i!)^\nu - \\
        & \frac{\lambda}{(i+1)^\nu} \lambda^{-(i+1)}((i+1)!)^\nu \bigg) \nabla f(x_{t-i-1}) \\
        = &\frac{C}{Z(\lambda,\nu)} \sum_{i=0}^{\infty} \frac{\lambda^i}{(i!)^\nu} \alpha \bigg( \frac{(i!)^\nu}{\lambda^{i}} - \frac{(i!)^\nu}{\lambda^{i}} \bigg) \nabla f(x_{t-i-1}) = 0 \\
    \end{align*}
\end{proof}

\begin{proof}[Proof of Theorem \ref{theorem:cmp_sigma_tune}]
    Let $\Psi(i) = \alpha (i) - \lambda \frac{\alpha (i+1)}{(i+1)^\nu}$, and hence
    \begin{align*}
        \Sigma_{p,\alpha}^\nabla
        &= \frac{1}{Z(\lambda,\nu)} \sum_{i=0}^{\infty} \frac{\lambda^i}{(i!)^\nu} \Psi(i) \nabla f(x_{t-i-1})
    \end{align*}
    
    Applying the adaptive step size (\ref{eq:adap_alpha_cmp_sigma_tune}) gives
    \begin{align*}
        \Psi(i) = \frac{i!^\nu}{\lambda^i} e^\lambda \alpha \left( c(i) - c(i+1) \right)
    \end{align*}
    
    Now,
    \begin{align*}
        &\Psi(i) = K
        \Leftrightarrow c(i) - c(i+1) = \frac{K}{\alpha e^{\lambda}}\frac{\lambda^i}{i!^\nu} \\
        & \ \ \Leftrightarrow c(i) = c(i-1) - \frac{K}{\alpha e^{\lambda}}\frac{\lambda^{i-1}}{(i-1)!^\nu} \\
        &= c(0) - \frac{K}{\alpha e^{\lambda}} \sum_{j=1}^i \frac{\lambda^{i-j}}{(i-j)!^\nu}
        = c(0) - \frac{K}{\alpha e^{\lambda}} \sum_{j=1}^i \frac{\lambda^{j}}{(j)!^\nu}
    \end{align*}
    
    Since $\alpha(0) = \alpha$, we have $c(0) = 1$. Now we have
    \begin{align*}
        \Sigma_{p,\alpha}^\nabla
        &= K \sum_{i=0}^{\infty} \frac{1}{Z(\lambda,\nu)} \frac{\lambda^i}{(i!)^\nu} \nabla f(x_{t-i-1}) \\
        &= K \mathbf{E} \left[ \nabla f(v_{t-1}) \right]
        = K \mathbf{E} \left[ x_t - x_{t-1} \right]
    \end{align*}
\end{proof}

\begin{proof}[Proof of Corollary \ref{cor:Pois_sigma_tune}]
    Under the Poisson $\tau$ model, which is $CMP$ with $\nu=1$, (\ref{eq:cmp_adap_factor}) rewrites to
    \begin{align*}
        c(i) &= 1 - \frac{K}{\alpha e^\lambda} \sum_{j=0}^{\tau-1} \frac{\lambda^j}{(j!)}
        = 1 - \frac{K}{\alpha} \frac{\Gamma(i,\lambda)}{(i-1)!} \\
        &= 1 - \frac{K}{\alpha} \frac{\Gamma(i,\lambda)}{\Gamma(i)}
    \end{align*}
\end{proof}

\begin{proof}[Proof of Theorem \ref{theorem:convex_convergence}]
    \begin{align*}
        \|&x_{t+1} - x^*\|^2
        = \|x_t - \alpha_t \nabla F(v_t) - x^*\|^2 \\
        &= \|x_t - x^*\|^2 + \alpha_t^2 \|\nabla F(v_t)\|^2 - 2\alpha_t(x_t - x^*)^T\nabla F(v_t) \\
        &= \|x_t - x^*\|^2 + \alpha_t^2 \|\nabla F(v_t)\|^2 - 2\alpha_t(x_t - x^*)^T \nabla F(x_t) \\
        & \ \ + 2\alpha_t(x_t - x^*)^T\big(\nabla F(x_t) - \nabla F(v_t)\big)
    \end{align*}
    
    Under expectation, conditioned on the natural filtration $\mathbb{F}_t^X = \big( (\tau_i)_{i=0}^{t}, \big(\nabla F(v_i)\big)_{i=0}^{t} \big)$ of the \textit{past} of the process, we have
    \begin{align*}
        &\mathbf{E}\big[\|x_{t+1} - x^*\|^2 \mid \tau_t, \mathbb{F}_{t-1}^X \big]
        = \|x_t - x^*\|^2 \\
        & \ \ - 2\alpha_t \mathbf{E}\big[ (x_t - x^*)^T \big( \nabla F(x_t) - \nabla F(x^*) \big) \big| \mathbb{F}_{t-1}^X \big] \\
        & \ \ + 2\alpha_t \mathbf{E}\big[ (x_t - x^*)^T\big(\nabla F(x_t) - \nabla F(v_t)\big) \big|  \mathbb{F}_{t-1}^X \big]
    \end{align*}
    
    Applying the assumptions (\ref{condition:convexity})-(\ref{condition:moment}) gives
    \begin{align*}
        &\mathbf{E}\big[\|x_{t+1} - x^*\|^2 \big| \tau_t, \mathbb{F}_{t-1}^X \big]
        \leq \|x_t - x^*\|^2 + M^2 \alpha_t^2 \\
        & \ \ - 2\alpha_t c \|x_t - x^*\|^2 + 2\alpha_t L \|x_t - x^*\| \|x_t - v_t\| \\
        &= (1-2c\alpha_t)\|x_t - x^*\|^2 + M^2 \alpha_t^2 \\
        & \ \ + 2\alpha_t L \|x_t - x^*\| \|x_t - v_t\| \\
        &= (1-2c\alpha_t)\|x_t - x^*\|^2 + M^2 \alpha_t^2 \\
        & \ \ + 2\alpha_t L \|x_t - x^*\| \sum_{i=1}^{\tau_t} x_{t-i+1} - x_{t-i}\| \\
        &\leq (1-2c\alpha_t)\|x_t - x^*\|^2 + M^2 \alpha_t^2 \\
        & \ \ + 2\alpha_t L \sum_{i=1}^{\tau_t} \|x_t - x^*\| \alpha_{t-i} \| \nabla F(v_{t-i}) \|
    \end{align*}
    
    The gradient process does not influence the expected delays, so we first consider the expectation conditioned on the gradient process $(\nabla)_0^t := \big(\nabla F(v_i)\big)_{i=0}^{t}$
    \begin{align*}
        &\mathbf{E}\left[\|x_{t+1} - x^*\|^2 \big| \tau_t, (\nabla)_0^t \right] \\
        & \ \ \leq (1-2c\alpha_t) \mathbf{E}\big[\|x_t - x^*\|^2 \big| \tau_t, (\nabla)_0^t \big] + M^2 \alpha_t^2 \\
        & \ \ + 2 L \alpha_t \sum_{i=1}^{\tau_t} \mathbf{E}\big[ \alpha_{t-i} \|x_t - x^*\| \big| \tau_t, (\nabla)_0^t \big] \| \nabla F(v_{t-i})\|
    \end{align*}
    From the non-anticipativity of the delay process we have
    \begin{align*}
        &\mathbf{E} \left[ \alpha_{t-i} \|x_t - x^*\| \mid \tau_t, (\nabla)_0^t \right] \\
        &= \mathbf{E} \left[ \mathbf{E} \left[ \alpha_{t-i} \|x_t - x^*\| \mid x_t \right] \mid \tau_t, (\nabla)_0^t \right] \\
        &= \mathbf{E} \left[ \|x_t - x^*\| \mathbf{E} \left[ \alpha_{t-i} \mid x_t \right] \mid \tau_t, (\nabla)_0^t \right] \\
        &= \mathbf{E} \left[ \alpha_{t-i} \right] \mathbf{E} \left[ \|x_t - x^*\| \mid (\nabla)_0^t \right]
    \end{align*}
    
    Since the delays and gradients are identically distributed we have $\mathbf{E}[\alpha_i] = \mathbf{E}[\alpha_j]$ for all $i,j$. Taking expectation conditioned on the last delay $\tau_t$ and applying Hölder's inequality gives
    \begin{align*}
        &\mathbf{E}\left[\|x_{t+1} - x^*\|^2 \mid \tau_t \right]
        \leq (1-2c\alpha_t) \mathbf{E}\big[\|x_t - x^*\|^2 \big] + M^2 \alpha_t^2 \\
        & \ \ + 2 L \tau_t \alpha_t \mathbf{E} \left[ \alpha_t \right] \sqrt{\mathbf{E} \left[ \|x_t - x^*\|^2 \right]} \sqrt{\mathbf{E} \left[ \| \nabla F(v_{t})\|^2 \right]}
    \end{align*}
    and the full expectation satisfies
    \begin{align*}
        &\mathbf{E}\left[\|x_{t+1} - x^*\|^2 \right] \leq (1-2c \mathbf{E}\left[\alpha_t\right] ) \mathbf{E}\big[\|x_t - x^*\|^2 \big] \\
        & \ \ + M^2 \mathbf{E}\left[\alpha_t^2\right] + 2 L M \mathbf{E}\left[ \tau_t \alpha_t \right] \mathbf{E} \left[ \alpha_{t} \right] \sqrt{\mathbf{E} \left[ \|x_t - x^*\|^2 \right]}
    \end{align*}
    
    As long as the process has not yet converged, i.e. $\mathbf{E}\big[\|x_t - x^*\|^2\big] > \epsilon$, we have
    \begin{align*}
        &\mathbf{E} \left[ \|x_{t+1} - x^*\|^2 \right]
        \leq \mathbf{E}\big[\|x_t - x^*\|^2 \big] ( 1-2c \mathbf{E}\left[\alpha_t\right] \\
        & \ \ + \epsilon^{-1}M^2 \mathbf{E}\left[\alpha_t^2\right] + 2 L M \epsilon^{1/2} \mathbf{E}\left[ \tau_t \alpha_t \right] \mathbf{E} \left[ \alpha_{t} \right] ) \\
        & \ \ =: \mathbf{E}\big[\|x_t - x^*\|^2] (1 - \delta) \\
        & \ \ \Rightarrow \mathbf{E}\left[\|x_T - x^*\|^2\right] \leq \mathbf{E}\left[\|x_0 - x^*\|^2\right] (1 - \delta)^T \\
        & \ \ \Rightarrow T \leq - \ln(1-\delta)^{-1} \ln \frac{\mathbf{E}\big[\|x_0 - x^*\|^2\big]}{\mathbf{E}\big[\|x_T - x^*\|^2\big]} \\
        & \ \ < \delta^{-1} \ln \left( \mathbf{E}\big[\|x_0 - x^*\|^2\big] \epsilon^{-1} \right)
    \end{align*}
    for any $T$ such that $\mathbf{E}\big[\|x_{T} - x^*\|^2] > \epsilon$. Equivalently, expected convergence is implied by $T$ exceeding the bound above, which concludes the proof.
\end{proof}

\begin{proof}[Proof of Corollary \ref{corollary:constant_alpha_convergence}]
    Let $\rho = \frac{c \epsilon M^{-1}}{M + 2 L \sqrt{\epsilon} \bar\tau}$. From Theorem \ref{theorem:convex_convergence} we have the improvement factor
    \begin{align*}
        \delta &= 2 \big( c - L M \epsilon^{1/2} \mathbf{E}\left[ \tau \alpha \right] \big) \mathbf{E}\left[\alpha\right] - \epsilon^{-1}M^2 \mathbf{E}\left[\alpha^2\right] \\
        &= 2 c \alpha -\epsilon^{-1}M \left( M + 2L\sqrt{\epsilon} \bar\tau \right) \alpha^2 \\
        &= c\rho^{-1}\alpha(2\rho - \alpha) \\
    \end{align*}
    so $\delta>0$ when $0<\alpha<2\rho$, and the improvement is maximized for $\theta=1$. Now, using the choice (\ref{eq:alpha_choice}) of step size, we have
    \begin{align*}
        \delta &= c \rho^{-1} \theta \rho(2\rho - \theta \rho) \\
        &= \theta(2 - \theta) c \rho
    \end{align*}
    Substituting for $\rho$, the convergence bound of Theorem \ref{theorem:convex_convergence} now rewrites to (\ref{eq:theta_convergence})
\end{proof}

\remove{

\begin{proof}[Proof of Corollary \ref{corollary:tau_adapt_concergence}]
    Using the choice of step size (\ref{eq:alpha_choice_linear_adaptive}) and the linearity of expectation, the improvement factor $\delta$ of Theorem \ref{theorem:convex_convergence} becomes
    \begin{align*}
        \delta &= 2 \big( c - L M \epsilon^{1/2} \mathbf{E}\left[ \tau \alpha \right] \big) \mathbf{E}\left[\alpha\right] - \epsilon^{-1}M^2 \mathbf{E}\left[\alpha^2\right] \\
        &= \theta \rho \bar\tau \left( 2 \big( c - LM\epsilon^{-1/2}\theta\rho \bar\tau \big) \mathbf{E}[\tau^{-1}] - \epsilon^{-1}M^2\theta \rho \bar\tau \mathbf{E}[\tau^{-2}] \right) \\
        &= \theta \rho \bar\tau \mathbf{E}\left[ 2 \big( c - LM\epsilon^{-1/2}\theta\rho \bar\tau \big) \tau^{-1} - \epsilon^{-1}M^2\theta \rho \bar\tau \tau^{-2}\right]
    \end{align*}
    
    It is easily seen that the restriction (\ref{eq:theta_restriction}) on $\theta$ ensures that the above expression is convex in $\tau$. Jensen's inequality gives
    \begin{align*}
        \delta &\geq \theta \rho \bar\tau \left( 2 \big( c - LM\epsilon^{-1/2}\theta\rho \bar\tau \big) \bar\tau^{-1} - \epsilon^{-1}M^2\theta \rho \bar\tau \bar\tau^{-2}\right) \\
        &= \theta \rho \left( 2c - \theta\rho \epsilon^{-1}M (M + 2LM\sqrt{\epsilon}\bar\tau ) \right) \\
        &= \theta \left( 2 - \theta \right) c \rho
    \end{align*}
\end{proof}

}

\begin{proof}[Proof of Corollary \ref{corollary:decaying_alpha_convergence}]
    Since $\alpha_t$ is a non-increasing function in $\tau_t$ we have
    \begin{align*}
        \mathbf{E}\left[ \tau_t \alpha(\tau_t) \right]
        &= \mathbf{E}\left[ \tau_t \alpha(\tau_t) \right] - \mathbf{E}\left[ \bar\tau \alpha(\tau_t) \right] + \mathbf{E}\left[ \tau_t \right] \mathbf{E}\left[ \alpha(\tau_t) \right] \\
        &= \mathbf{E}\left[ (\tau_t - \bar\tau)(\alpha(\tau_t) - \alpha(\bar\tau)) \right] + \mathbf{E}\left[ \tau \right] \mathbf{E}\left[ \alpha \right] \\
        &\leq \mathbf{E}\left[ \tau \right] \mathbf{E}\left[ \alpha \right]
    \end{align*}
    
    Using this property, (\ref{eq:convex_convergence}) rewrites to (\ref{eq:decaying_alpha_convergence})
\end{proof}

\end{document}